\documentclass[12pt,reqno]{article}
\usepackage[usenames]{color}
\usepackage{amssymb}
\usepackage{graphicx}
\usepackage{amscd}

\usepackage{graphics,amsmath,mathtools,amssymb,relsize}
\usepackage{bm}
\usepackage{amsthm}
\usepackage{amsfonts}
\usepackage{latexsym}
\usepackage{epsf}

\newtheorem{theorem}{Theorem}[section]
\newtheorem{corollary}[theorem]{Corollary}
\newtheorem{lemma}[theorem]{Lemma}
\newtheorem{proposition}[theorem]{Proposition}

\newtheorem{defin}[theorem]{Definition}
\newenvironment{definition}{\begin{defin}\normalfont\quad}{\end{defin}}
\newtheorem{examp}[theorem]{Example}

\newtheorem{rema}[theorem]{Remark}
\newtheorem{prob}[theorem]{Problem}

\numberwithin{equation}{section}

\newcommand{\bt}{\begin{thm}}
\newcommand{\et}{\end{thm}}
\newcommand{\bp}{\begin{proof}}
\newcommand{\ep}{\end{proof}}
\newcommand{\bprop}{\begin{prop}}
\newcommand{\eprop}{\end{prop}}
\newcommand{\bl}{\begin{lemma}}
\newcommand{\el}{\end{lemma}}
\newcommand{\bc}{\begin{corollary}}
\newcommand{\ec}{\end{corollary}}
\newcommand{\Z}{\mathbb{Z}}
\newcommand{\C}{\mathbb{C}}
\newcommand{\be}{\begin{enumerate}}
\newcommand{\ee}{\end{enumerate}}

\newcommand{\OMIT}[1]{}

\title{Unweighted linear congruences with distinct coordinates \\ and the Varshamov--Tenengolts codes}\author{Khodakhast Bibak \thanks{Department of Computer Science, University of Victoria, Victoria, BC, Canada V8W 3P6. Email: {\tt \{kbibak,bmkapron,srinivas\}@uvic.ca}} \and Bruce M. Kapron \footnotemark[1] \and Venkatesh Srinivasan \footnotemark[1]}

\begin{document}

\maketitle

\begin{abstract}
In this paper, we first give explicit formulas for the number of solutions of unweighted linear congruences with distinct coordinates. Our main tools are properties of Ramanujan sums and of the discrete Fourier transform of arithmetic functions. Then, as an application, we derive an explicit formula for the number of codewords in the Varshamov--Tenengolts code $VT_b(n)$ with Hamming weight $k$, that is, with exactly $k$ $1$'s. The Varshamov--Tenengolts codes are an important class of codes that are capable of correcting asymmetric errors on a $Z$-channel. As another application, we derive Ginzburg's formula for the number of codewords in $VT_b(n)$, that is, $|VT_b(n)|$. We even go further and discuss connections to several other combinatorial problems, some of which have appeared in seemingly unrelated contexts. This provides a general framework and gives new insight into all these problems which might lead to further work.
\end{abstract}


\section{Introduction}\label{Sec_1}

A \textit{$Z$-channel} (also called a \textit{binary asymmetric channel}) is a channel with binary input and binary output where a transmitted $0$ is always received correctly but a transmitted $1$ may be received as either $1$ or $0$. These channels have many applications, for example, some data storage systems and optical communication systems can be modelled using these channels. In 1965, Varshamov and Tenengolts \cite{VATE} introduced an important class of codes, known as the Varshamov--Tenengolts codes or VT-codes, that are capable of correcting asymmetric errors on a $Z$-channel (see also \cite{STYO, VAR}). Levenshtein \cite{LEV1, LEV2}, by giving an elegant decoding algorithm, showed that these codes could also be used for correcting a single deletion or insertion. Using the Varshamov--Tenengolts codes, Gevorkyan and Kabatiansky \cite{GEKA} constructed a class of binary codes of a specific length correcting single localized errors whose cardinality attains the ordinary Hamming bound.    

\begin{definition}
Let $n$ be a positive integer and $0\leq b\leq n$ be a fixed integer. The Varshamov--Tenengolts code $VT_b(n)$ is the set of all binary vectors $\langle y_1,\ldots,y_n \rangle$ such that 
$$
\sum_{i=1}^{n}iy_i \equiv b \pmod{n+1}.
$$ 
\end{definition}
For example, $VT_0(5)=\lbrace 00000, 10001, 01010, 11100, 00111, 11011 \rbrace$, where we have shown vectors as strings. So, $|VT_0(5)|=6$. The \textit{Hamming weight} of a string over an alphabet is defined as the number of non-zero symbols in the string. Equivalently, the Hamming weight of a string is the Hamming distance between that string and the all-zero string of the same length. For example, the Hamming weight of $01010$ is $2$, and the number of codewords in $VT_0(5)$ with Hamming weight $2$ is $2$.

Varshamov in his fundamental paper ``On an arithmetic function with an application in the theory of coding" (\cite{VAR2}) proved that the maximum number of codewords in the Varshamov--Tenengolts code $VT_b(n)$ is achieved when $b=0$, that is, $|VT_0(n)| \geq |VT_b(n)|$ for all $b$. Several natural questions arise: What is the number of codewords in the Varshamov--Tenengolts code $VT_b(n)$, that is, $|VT_b(n)|$? Given a positive integer $k$, what is the number of codewords in $VT_b(n)$ with Hamming weight $k$, that is, with exactly $k$ $1$'s? Ginzburg \cite{GIN} in 1967 considered the first question and proved an explicit formula for $|VT_b(n)|$. In this paper, we deal with both questions and obtain explicit formulas for them via a novel approach, namely, \textit{connecting the Varshamov--Tenengolts codes to linear congruences with distinct coordinates}. We even go further and show that the number of solutions of these congruences is related to several other combinatorial problems, some of which have appeared in seemingly unrelated contexts. (For example, as we will discuss in Section~\ref{Sec_4}, Razen, Seberry, and Wehrhahn \cite{RSW} considered two special cases of a function considered in this paper and gave an application in coding theory in finding the complete weight enumerator of a code generated by a circulant matrix.) This provides a general framework and gives new insight into all these problems which might lead to further work. Let us now describe these congruences.

Throughout the paper, we use $(a_1,\ldots,a_k)$ to denote the greatest common divisor (gcd) of the integers $a_1,\ldots,a_k$, and write $\langle a_1,\ldots,a_k\rangle$ for an ordered $k$-tuple of integers. Let $a_1,\ldots,a_k,b,n\in \Z$, $n\geq 1$. A linear congruence in $k$ unknowns $x_1,\ldots,x_k$ is of the form
\begin{align} \label{cong form}
a_1x_1+\cdots +a_kx_k\equiv b \pmod{n}.
\end{align}

By a solution of (\ref{cong form}), we mean an $\mathbf{x}=\langle x_1,\ldots,x_k \rangle \in \mathbb{Z}_n^k$ that satisfies (\ref{cong form}). The following result, proved by D. N. Lehmer \cite{LEH2}, gives the number of solutions of the above linear congruence:

\begin{proposition}\label{Prop: lin cong}
Let $a_1,\ldots,a_k,b,n\in \Z$, $n\geq 1$. The linear congruence $a_1x_1+\cdots +a_kx_k\equiv b \pmod{n}$ has a solution $\langle x_1,\ldots,x_k \rangle \in \Z_{n}^k$ if and only if $\ell \mid b$, where
$\ell=(a_1, \ldots, a_k, n)$. Furthermore, if this condition is satisfied, then there are $\ell n^{k-1}$ solutions.
\end{proposition}

Counting the number of solutions of the above congruence with some restrictions on the solutions is also a problem of great interest. As an important example, one can mention the restrictions $(x_i,n)=t_i$ ($1\leq i\leq k$), where $t_1,\ldots,t_k$ are given positive divisors of $n$. The number of solutions of the linear congruences with the above restrictions, which we called {\it restricted linear congruences} in \cite{BKSTT}, was first considered by Rademacher \cite{Rad1925} in 1925 and Brauer \cite{Bra1926} in 1926, in the special case of $a_i=t_i=1$ $(1\leq i \leq k)$. Since then, this problem has been studied, in several other special cases, in many papers (very recently, we studied it in its `most general case' in \cite{BKSTT}) and has found very interesting applications in number theory, combinatorics, geometry, physics, computer science, and cryptography; see \cite{BKS2, BKSTT, BKSTT2, JAWILL} for a detailed discussion about this problem and a comprehensive list of references. 

Another restriction of potential interest is imposing the condition that all $x_i$ are {\it distinct} modulo $n$. Unlike the first problem, there seems to be very little published on the second problem. Recently, Grynkiewicz et al. \cite{GPP}, using tools from additive combinatorics and group theory, proved necessary and sufficient conditions under which the linear congruence $a_1x_1+\cdots +a_kx_k\equiv b \pmod{n}$, where $a_1,\ldots,a_k,b,n$ ($n\geq 1$) are arbitrary integers, has a solution $\langle x_1,\ldots,x_k \rangle \in \Z_{n}^k$ with all $x_i$ distinct modulo $n$; see also \cite{ADP, GPP} for connections to zero-sum theory. So, it would be an interesting problem to give an explicit formula for the number of such solutions. Quite surprisingly, this problem was first considered, in a special case, by Sch\"{o}nemann \cite{SCH} almost two centuries ago(!) but his result seems to have been forgotten. Sch\"{o}nemann \cite{SCH} proved an explicit formula for the number of such solutions when $b=0$, $n=p$ a prime, and $\sum_{i=1}^k a_i \equiv 0 \pmod{p}$ but $\sum_{i \in I} a_i \not\equiv 0 \pmod{p}$ for all $I\varsubsetneq \lbrace 1, \ldots, k\rbrace$. Very recently, the authors \cite{BKS6} generalized Sch\"{o}nemann's theorem using Proposition~\ref{Prop: lin cong} and a result on graph enumeration recently obtained by Ardila et al. \cite{ACH}. Specifically, we obtained an explicit formula for the number of solutions of the linear congruence $a_1x_1+\cdots +a_kx_k\equiv b \pmod{n}$, with all $x_i$ distinct modulo $n$, when $(\sum_{i \in I} a_i, n)=1$ for all $I\varsubsetneq \lbrace 1, \ldots, k\rbrace$, where $a_1,\ldots,a_k,b,n$ $(n\geq 1)$ are arbitrary integers. Clearly, this result does not resolve the problem in its full generality; for example, it does not cover the important case of $a_i=1$ ($1\leq i\leq k$) and this is what we consider in this paper with an entirely different approach. Specifically, we give an explicit formula for the number $N_n(k,b)$ of such solutions when $a_i=1$ ($1\leq i\leq k$), and do the same when in addition all $x_i$ are \textit{positive} modulo $n$. 

Our main tools in this paper are properties of Ramanujan sums and of the discrete Fourier transform of arithmetic functions which are reviewed in the next section. In Section~\ref{Sec_3}, we derive the explicit formulas, and discuss applications to the Varshamov--Tenengolts codes. In Section~\ref{Sec_4}, we discuss connections to several other combinatorial contexts.

\section{Ramanujan sums and discrete Fourier transform} \label{Sec_2} 

Let $e(x)=\exp(2\pi ix)$ be the complex exponential with period 1. For integers $m$ and $n$ with $n \geq 1$ the quantity 
\begin{align}\label{def1}
c_n(m) = \mathlarger{\sum}_{\substack{j=1 \\ (j,n)=1}}^{n}
e\!\left(\frac{jm}{n}\right)
\end{align}
is called a {\it Ramanujan sum}. It is the sum of the $m$-th powers of the primitive $n$-th roots of unity, and is also denoted by $c(m,n)$ in the literature. From (\ref{def1}), it is clear that $c_n(-m) = c_n(m)$. Clearly, $c_n(0)=\varphi (n)$, where $\varphi (n)$ is {\it Euler's totient function}. Also, $c_n(1)=\mu (n)$, where $\mu (n)$ is the {\it M\"{o}bius function}. The following theorem, attributed to Kluyver~\cite{KLU}, gives an explicit formula for $c_n(m)$:

\begin{theorem} \label{thm:Ram Mob}
For integers $m$ and $n$, with $n \geq 1$,
\begin{align}\label{for:Ram Mob}
c_n(m) = \mathlarger{\sum}_{d\, \mid\, (m,n)} \mu
\!\left(\frac{n}{d}\right)d.
\end{align}
\end{theorem}

By applying the M\"{o}bius inversion formula, Theorem~\ref{thm:Ram Mob} yields the following property: For $m,n\geq 1$,
\begin{align} \label{Orth1 for cons}
\sum_{d\, \mid\, n} c_{d}(m)&=
  \begin{cases}
    n, & \text{if $n\mid m$},\\
    0, & \text{if $n\nmid m$}.
  \end{cases}
\end{align}

The {\it von Sterneck number} (\cite{von}) is defined by 

\begin{align}\label{def3}
\Phi(m,n)=\frac{\varphi (n)}{\varphi \left(\frac{n}{\left(m,n\right)}\right)}\mu \!\left(\frac{n}{\left(m,n\right)} \right).
\end{align}

A crucial fact in studying Ramanujan sums and their applications is that they coincide with the von Sterneck number. This result is attributed to Kluyver~\cite{KLU}:

\begin{theorem} \label{thm:von rama}
For integers $m$ and $n$, with $n \geq 1$, we have
\begin{align}\label{von rama for}
\Phi(m,n)=c_n(m).
\end{align}
\end{theorem}

A function $f:\Z \to \C$ is called {\it periodic} with period $n$ (also called {\it $n$-periodic} or {\it periodic} modulo $n$) if $f(m + n) = f(m)$, for every $m\in \mathbb{Z}$. In this case $f$ is determined by the finite vector $(f(1),\ldots,f(n))$. From (\ref{def1}) it is clear that $c_n(m)$ is a periodic function of $m$ with period $n$.

We define the {\it discrete Fourier transform} (DFT) of an $n$-periodic function $f$ as the function 
$\widehat{f}={\cal F}(f)$, given by
\begin{align}\label{FFT1}
\widehat{f}(b)=\mathlarger{\sum}_{j=1}^{n}f(j)e\! \left(\frac{-bj}{n}\right)\quad (b\in \Z).
\end{align}

The standard representation of $f$ is obtained from the Fourier representation $\widehat{f}$ by
\begin{align}\label{FFT2}
f(b)=\frac1{n} \mathlarger{\sum}_{j=1}^{n}\widehat{f}(j)e\!\left(\frac{bj}{n}\right) \quad (b\in \Z),
\end{align}
which is the {\it inverse discrete Fourier transform} (IDFT); see, e.g., \cite[p.\ 109]{MOVA}.

\section{Solutions with distinct coordinates}\label{Sec_3}

In this section, we obtain an explicit formula for the number of solutions $\langle x_1,\ldots,x_k \rangle \in \Z_{n}^k$ of the linear congruence $x_1+\cdots +x_k\equiv b \pmod{n}$, with all $x_i$ distinct modulo $n$. First, we need some preliminary results.

\begin{lemma}\label{lem: cyclo 1}
Let $n$ be a positive integer and $m$ be a non-negative integer. We have 

$$
\mathlarger{\prod}_{j=1}^{n}\left(1-ze^{2\pi ijm/n}\right)=(1-z^{\frac{n}{d}})^d,
$$
where $d=(m,n)$.
\end{lemma}

\begin{proof}
It is well-known that (see, e.g., \cite[p. 167]{STAN})

$$
1-z^n=\mathlarger{\prod}_{j=1}^{n}\left(1-ze^{2\pi ij/n}\right).
$$
Now, letting $d=(m,n)$, we obtain

\begin{align*}
\mathlarger{\prod}_{j=1}^{n}\left(1-ze^{2\pi ijm/n}\right) &= \mathlarger{\prod}_{j=1}^{n} \left(1-ze^{2\pi ij\frac{m/d}{n/d}}\right)\\
&= \left(\mathlarger{\prod}_{j=1}^{n/d}\left(1-ze^{2\pi ij\frac{m/d}{n/d}}\right)\right)^d\\
&{\stackrel{(\frac{m}{d},\frac{n}{d})=1}{=}} \left(\mathlarger{\prod}_{j=1}^{n/d}\left(1-ze^{\frac{2\pi ij}{n/d}}\right)\right)^d=(1-z^{\frac{n}{d}})^d.
\end{align*}
\end{proof}

Similarly, we can prove that:

\begin{lemma}\label{lem: cyclo 2}
Let $n$ be a positive integer and $m$ be a non-negative integer. We have 

$$
\mathlarger{\prod}_{j=1}^{n}\left(z-e^{2\pi ijm/n}\right)=(z^{\frac{n}{d}}-1)^d,
$$
where $d=(m,n)$.
\end{lemma}

Now, we simply get:

\begin{corollary}\label{cor: cyclo 1}
Let $n$ be a positive integer and $m$, $k$ be non-negative integers. The coefficient of $z^k$ in

$$
\mathlarger{\prod}_{j=1}^{n}\left(1+ze^{2\pi ijm/n}\right),
$$
is $(-1)^{k+\frac{kd}{n}}\binom{d}{\frac{kd}{n}}$, where $d=(m,n)$. Note that the binomial coefficient $\binom{d}{\frac{kd}{n}}$ equals zero if $\frac{kd}{n}$ is not an integer.
\end{corollary}

Now, we are ready to obtain an explicit formula for the number of solutions of the linear congruence.

\begin{theorem} \label{main thm dist ai=1}
Let $n$ be a positive integer and $b \in \Z_n$. The number $N_n(k,b)$ of solutions $\langle x_1,\ldots,x_k \rangle \in \Z_{n}^k$ of the linear congruence $x_1+\cdots +x_k\equiv b \pmod{n}$, with all $x_i$ distinct modulo $n$, is
\begin{align} \label{main thm dist ai=1: for}
N_n(k,b)=\frac{(-1)^k k!}{n}\mathlarger{\sum}_{d\, \mid \, (n,\;k)}(-1)^{\frac{k}{d}}c_{d}(b)\binom{\frac{n}{d}}{\frac{k}{d}}.
\end{align}
\end{theorem}

\begin{proof}
It is well-known that (see, e.g., \cite[pp. 3-4]{GUP}) the number of partitions of $b$ into exactly $k$ \textit{distinct} parts each taken from the given set $A$, is the coefficient of $q^bz^k$ in 

$$
\mathlarger{\prod}_{j \in A}\left(1+zq^j\right).
$$
Now, take $A=\Z_n$ and $q=e^{2\pi im/n}$, where $m$ is a non-negative integer. Then, the number $P_n(k,b)$ of partitions of $b$ into exactly $k$ \textit{distinct} parts each taken from $\Z_n$ (that is, the number of solutions of the above linear congruence, with all $x_i$ distinct modulo $n$, if order does not matter), is the coefficient of $e^{2\pi ibm/n}z^k$ in 

$$
\mathlarger{\prod}_{j=1}^{n}\left(1+ze^{2\pi ijm/n}\right).
$$
This in turn implies that

$$
\mathlarger{\sum}_{b=1}^{n}P_n(k,b)e^{2\pi ibm/n} = \text{the coefficient of $z^k$ in $\mathlarger{\prod}_{j=1}^{n}\left(1+ze^{2\pi ijm/n}\right)$}.
$$
Let $e(x)=\exp(2\pi ix)$. Note that $N_n(k,b)=k!P_n(k,b)$. Now, using Corollary~\ref{cor: cyclo 1}, we get

$$
\mathlarger{\sum}_{b=1}^{n}N_n(k,b)e\left(\frac{bm}{n}\right) = (-1)^{k+\frac{kd}{n}}k!\binom{d}{\frac{kd}{n}},
$$
where $d=(m,n)$. Now, by (\ref{FFT1}) and (\ref{FFT2}), we obtain
\begin{align*}
N_n(k,b) &= \frac{(-1)^{k}k!}{n}\mathlarger{\sum}_{m=1}^{n}(-1)^{\frac{kd}{n}}e\left(\frac{-bm}{n}\right)\binom{d}{\frac{kd}{n}}\\
&= \frac{(-1)^{k}k!}{n}\mathlarger{\sum}_{d\, \mid \, n}\mathlarger{\sum}_{\substack{m=1 \\ (m,\;n)=d}}^{n}(-1)^{\frac{kd}{n}}e\left(\frac{-bm}{n}\right)\binom{d}{\frac{kd}{n}}\\
&{\stackrel{m'=m/d}{=}} \;\; \frac{(-1)^{k}k!}{n}\mathlarger{\sum}_{d\, \mid \, n}\mathlarger{\sum}_{\substack{m'=1 \\ (m',\;n/d)=1}}^{n/d}(-1)^{\frac{kd}{n}}e\left(\frac{-bm'}{n/d}\right)\binom{d}{\frac{kd}{n}}\\
&= \frac{(-1)^{k}k!}{n}\mathlarger{\sum}_{d\, \mid \, n}(-1)^{\frac{kd}{n}}c_{n/d}(-b)\binom{d}{\frac{kd}{n}}\\
&= \frac{(-1)^{k}k!}{n}\mathlarger{\sum}_{d\, \mid \, n}(-1)^{\frac{kd}{n}}c_{n/d}(b)\binom{d}{\frac{kd}{n}}\\
&= \frac{(-1)^{k}k!}{n}\mathlarger{\sum}_{d\, \mid \, n}(-1)^{\frac{k}{d}}c_{d}(b)\binom{\frac{n}{d}}{\frac{k}{d}}\\
&= \frac{(-1)^{k}k!}{n}\mathlarger{\sum}_{d\, \mid \, (n,\;k)}(-1)^{\frac{k}{d}}c_{d}(b)\binom{\frac{n}{d}}{\frac{k}{d}}.
\end{align*}
\end{proof}

\begin{corollary} \label{special cases: b=0,1}
If $n$ or $k$ is odd then from (\ref{main thm dist ai=1: for}) we obtain the following important special cases of the function $P_n(k,b)=\frac{1}{k!}N_n(k,b)$:
\begin{align} \label{special case: b=0}
P_n(k,0)= \frac{1}{n}\mathlarger{\sum}_{d\, \mid \, (n,\;k)}\varphi(d)\binom{\frac{n}{d}}{\frac{k}{d}},
\end{align}
\begin{align} \label{special case: b=1}
P_n(k,1)= \frac{1}{n}\mathlarger{\sum}_{d\, \mid \, (n,\;k)}\mu(d)\binom{\frac{n}{d}}{\frac{k}{d}}.
\end{align}
\end{corollary}

\begin{corollary}
If $(n,k)=1$ then (\ref{main thm dist ai=1: for}) is independent of $b$ and simplifies as 
$$N_n(k)=\frac{k!}{n}\binom{n}{k}.$$
(Of course, this can also be proved directly.) If in addition we have $n=2k+1$ then 
$$
P_n(k)=\frac{1}{k!}N_n(k)=\frac{1}{2k+1}\binom{2k+1}{k}=\frac{1}{k+1}\binom{2k}{k},
$$
which is the Catalan number.
\end{corollary}

\begin{rema}
Using (\ref{Orth1 for cons}), it is easy to see that (\ref{main thm dist ai=1: for}) also works when $k=0$.
\end{rema}

Now, we introduce the important function $T_n(b)$ which is the sum of $P_n(k,b)$ over $k$. There are several interpretations for the function $T_n(b)$, for example, $T_n(b)$ can be interpreted as the number of subsets of the set $\lbrace 1, 2, \ldots, n \rbrace$ which sum to $b$ modulo $n$. 

\begin{corollary}\label{nice function}
Let $T_n(b):=\sum_{k=0}^{n}\frac{1}{k!}N_n(k,b)=\sum_{k=0}^{n}P_n(k,b)$. Then we have 
\begin{align} \label{nice function: for}
T_n(b)=\frac{1}{n}\mathlarger{\sum}_{\substack{d\, \mid \, n \\ d \; \textnormal{odd}}}c_{d}(b)2^{\frac{n}{d}}.
\end{align}
\end{corollary}

\begin{proof} We have
\begin{align*}
T_n(b)&= \mathlarger{\sum}_{k=0}^{n} \frac{(-1)^{k}}{n}\mathlarger{\sum}_{d\, \mid \, (n,\;k)}(-1)^{\frac{k}{d}}\binom{\frac{n}{d}}{\frac{k}{d}}c_{d}(b)\\
&= \frac{1}{n}\mathlarger{\sum}_{d\, \mid \, n}c_{d}(b)\mathlarger{\sum}_{\substack{k=0 \\ d\, \mid \,k}}^{n}(-1)^{k+\frac{k}{d}}\binom{\frac{n}{d}}{\frac{k}{d}}\\
&= \frac{1}{n}\mathlarger{\sum}_{\substack{d\, \mid \, n \\ d \; \textnormal{odd}}}c_{d}(b)\mathlarger{\sum}_{\substack{k=0 \\ d\, \mid \,k}}^{n}(-1)^{k+\frac{k}{d}}\binom{\frac{n}{d}}{\frac{k}{d}}
+\frac{1}{n}\mathlarger{\sum}_{\substack{d\, \mid \, n \\ d \; \textnormal{even}}}c_{d}(b)\mathlarger{\sum}_{\substack{k=0 \\ d\, \mid \,k}}^{n}(-1)^{k+\frac{k}{d}}\binom{\frac{n}{d}}{\frac{k}{d}}
\\
&=\frac{1}{n}\mathlarger{\sum}_{\substack{d\, \mid \, n \\ d \; \textnormal{odd}}}c_{d}(b)2^{\frac{n}{d}}.
\end{align*}
Note that in the last equality we have used the fact that if $d \mid n$ and $d$ is even then
$$
\mathlarger{\sum}_{\substack{k=0 \\ d\, \mid \,k}}^{n}(-1)^{k+\frac{k}{d}}\binom{\frac{n}{d}}{\frac{k}{d}}
= \mathlarger{\sum}_{\substack{k=0 \\ d\, \mid \,k}}^{n}(-1)^{\frac{k}{d}}\binom{\frac{n}{d}}{\frac{k}{d}}= 0.
$$
\end{proof}

What is the number of subsets of the set $\lbrace 1, 2, \ldots, n-1 \rbrace$ which sum to $b$ modulo $n$? Using Corollary~\ref{nice function}, we can obtain an explicit formula for the number of such subsets (see also \cite{MAZ}).

\begin{corollary}\label{nice function 2}
The number $T'_n(b)$ of subsets of the set $\lbrace 1, 2, \ldots, n-1 \rbrace$ which sum to $b$ modulo $n$ is  
\begin{align} \label{nice function 2: for}
T'_n(b)=\frac{1}{2}T_n(b)=\frac{1}{2n}\mathlarger{\sum}_{\substack{d\, \mid \, n \\ d \; \textnormal{odd}}}c_{d}(b)2^{\frac{n}{d}}.
\end{align}
\end{corollary}

\begin{proof}
Let $A$ be a subset of the set $\lbrace 1, 2, \ldots, n-1 \rbrace$ which sum to $b$ modulo $n$. Then $A$ and $A\cup \lbrace n \rbrace$ are both subsets of the set $\lbrace 1, 2, \ldots, n \rbrace$ and both sum to $b$ modulo $n$. Therefore, $T'_n(b)=\frac{1}{2}T_n(b)$.
\end{proof}

Ginzburg \cite{GIN} in 1967 proved an explicit formula for the number of codewords in the $q$-ary Varshamov--Tenengolts codes, where $q$ is an arbitrary positive integer. This result was later rediscovered by Stanley and Yoder \cite{STYO} in 1973, and in the binary case (that is, when $q=2$) by Sloane \cite{SLO} in 2002. Now, we give a short proof for the binary case which we derive as a consequence of our results. 

\begin{corollary}\label{VT exa tot}
The number $|VT_b(n)|$ of codewords in the Varshamov--Tenengolts code $VT_b(n)$ is
\begin{align} \label{VT exa tot: for}
|VT_b(n)|=\frac{1}{2(n+1)}\mathlarger{\sum}_{\substack{d\, \mid \, n+1 \\ d \; \textnormal{odd}}}c_{d}(b)2^{\frac{n+1}{d}}.
\end{align}
\end{corollary}

\begin{proof}
Let $\langle y_1,\ldots,y_n \rangle$ be a codeword in $VT_b(n)$. Note that $\sum_{i=1}^{n}iy_i$ is just the sum of some elements of the set $\lbrace 1, 2, \ldots, n \rbrace$. Therefore, finding the number of codewords in $VT_b(n)$ boils down to finding the number of subsets of the set $\lbrace 1, 2, \ldots, n \rbrace$ which sum to $b$ modulo $n+1$. The result now follows by a direct application of Corollary~\ref{nice function 2}.
\end{proof}

In some applications (for example, in coding theory) we also need to consider the case that all $x_i$ are \textit{positive} and \textit{distinct} modulo $n$. Now, we obtain an explicit formula for the number of such solutions.

\begin{theorem} \label{main thm dist pos ai=1}
Let $n$ be a positive integer and $b \in \Z_n$. The number $N_n^{>0}(k,b)$ of solutions $\langle x_1,\ldots,x_k \rangle \in \Z_{n}^k$ of the linear congruence $x_1+\cdots +x_k\equiv b \pmod{n}$, with all $x_i$ positive and distinct modulo $n$, is
\begin{align} \label{main thm dist pos ai=1: for}
N_n^{>0}(k,b)=\frac{(-1)^k k!}{n}\mathlarger{\sum}_{d\, \mid \, n}(-1)^{\lfloor\frac{k}{d}\rfloor}c_{d}(b)\binom{\frac{n}{d}-1}{\lfloor\frac{k}{d}\rfloor}.
\end{align}
\end{theorem}

\begin{proof}
Clearly, $N_n^{>0}(k,b)=N_n(k,b)-N_n^{0}(k,b)$, where $N_n^{0}(k,b)$ denotes the number of solutions $\langle x_1,\ldots,x_k \rangle \in \Z_{n}^k$ with all $x_i$ distinct modulo $n$ and one of $x_i$ is zero modulo $n$. Also, clearly, $N_n^{0}(k,b)=kN_n^{>0}(k-1,b)$. Thus, 
\begin{align}\label{relation bet N and N>0}
N_n(k,b)=N_n^{>0}(k,b)+kN_n^{>0}(k-1,b).
\end{align}
Now, we use Theorem~\ref{main thm dist ai=1}. We have
\begin{align*}
N_n(k,b)&= \frac{(-1)^k k!}{n}\mathlarger{\sum}_{d\, \mid \, (n,\;k)}(-1)^{\frac{k}{d}}c_{d}(b)\binom{\frac{n}{d}}{\frac{k}{d}}\\
&= \frac{(-1)^k k!}{n}\mathlarger{\sum}_{d\, \mid \, (n,\;k)}(-1)^{\frac{k}{d}}c_{d}(b)\left(\binom{\frac{n}{d}-1}{\frac{k}{d}}+\binom{\frac{n}{d}-1}{\frac{k}{d}-1}\right)\\
&= \frac{(-1)^k k!}{n}\mathlarger{\sum}_{d\, \mid \, n}c_{d}(b)\left((-1)^{\frac{k}{d}}\binom{\frac{n}{d}-1}{\frac{k}{d}}-(-1)^{\frac{k}{d}-1}\binom{\frac{n}{d}-1}{\frac{k}{d}-1}\right)\\
&= \frac{(-1)^k k!}{n}\mathlarger{\sum}_{d\, \mid \, n}c_{d}(b)\left((-1)^{\lfloor\frac{k}{d}\rfloor}\binom{\frac{n}{d}-1}{\lfloor\frac{k}{d}\rfloor}-(-1)^{\lfloor\frac{k-1}{d}\rfloor}\binom{\frac{n}{d}-1}{\lfloor\frac{k-1}{d}\rfloor}\right)\\
&= \frac{(-1)^k k!}{n}\mathlarger{\sum}_{d\, \mid \, n}(-1)^{\lfloor\frac{k}{d}\rfloor}c_{d}(b)\binom{\frac{n}{d}-1}{\lfloor\frac{k}{d}\rfloor}\\
&+k\frac{(-1)^{k-1} (k-1)!}{n}\mathlarger{\sum}_{d\, \mid \, n}(-1)^{\lfloor\frac{k-1}{d}\rfloor}c_{d}(b)\binom{\frac{n}{d}-1}{\lfloor\frac{k-1}{d}\rfloor}.
\end{align*}
Note that in the fourth equality above we have used the fact that $\lfloor\frac{k}{d}\rfloor=\lfloor\frac{k-1}{d}\rfloor+1$ if $d \mid k$, and $\lfloor\frac{k}{d}\rfloor=\lfloor\frac{k-1}{d}\rfloor$ if $d\nmid k$. Now, recalling (\ref{relation bet N and N>0}) we obtain the desired result.
\end{proof}

We believe that Theorem~\ref{main thm dist pos ai=1} is also a strong tool and might lead to interesting applications. Denote by $VT_b^{1,k}(n)$ the set of codewords in the Varshamov--Tenengolts code $VT_b(n)$ with Hamming weight $k$. Theorem~\ref{main thm dist pos ai=1} immediately gives an explicit formula for the number of such codewords. This result is useful in the study of a class of binary codes that are immune to single repetitions \cite{DOAN}.  

\begin{corollary}\label{VT exa k1s}
The number $|VT_b^{1,k}(n)|$ of codewords in the Varshamov--Tenengolts code $VT_b(n)$ with Hamming weight $k$ is
\begin{align} \label{VT exa k1s: for}
|VT_b^{1,k}(n)|=\frac{(-1)^k}{n+1}\mathlarger{\sum}_{d\, \mid \, n+1}(-1)^{\lfloor\frac{k}{d}\rfloor}c_{d}(b)\binom{\frac{n+1}{d}-1}{\lfloor\frac{k}{d}\rfloor}.
\end{align}
\end{corollary}

\begin{proof}
Let $\langle y_1,\ldots,y_n \rangle$ be a codeword in $VT_b(n)$ with Hamming weight $k$, that is, with exactly $k$ $1$'s. Denote by $x_j$ the position of the $j$th one. Note that $1\leq j \leq k$ and $1 \leq x_1 < x_2 < \cdots < x_k \leq n$. Now, we have 
$$
\sum_{i=1}^{n}iy_i \equiv b \pmod{n+1} \Longleftrightarrow x_1+\cdots +x_k\equiv b \pmod{n+1}.
$$
Therefore, finding the number of codewords in $VT_b(n)$ with Hamming weight $k$ boils down to finding the number of solutions $\langle x_1,\ldots,x_k \rangle \in \Z_{n+1}^k$ of the linear congruence $x_1+\cdots +x_k\equiv b \pmod{n+1}$, with all $x_j$ positive and distinct modulo $n+1$, and with disregarding the order of the coordinates. The result now follows by a direct application of Theorem~\ref{main thm dist pos ai=1}.
\end{proof}

\begin{rema}
There is an earlier interesting result of Dolecek and Anantharam \cite{DOAN} which gives the formula (\ref{VT exa k1s: for}) in a special case where the Hamming weight is dependent on the modulus, but here we give a more general treatment where the Hamming weight is \textit{arbitrary}. Of course, the expression (3.7) in their paper is exactly the same as our formula (\ref{main thm dist ai=1: for}), so it is an interesting problem to prove a 1-1 correspondence between these two results.
\end{rema}

\section{More connections}\label{Sec_4}

Interestingly, some special cases of the functions $P_n(k,b)$, $N_n(k,b)$, $T_n(b)$, and $T'_n(b)$ that we studied in this paper have appeared in a wide range of combinatorial problems, sometimes in seemingly unrelated contexts. Here we briefly mention some of these connections. It would be interesting to prove 1-1 correspondences between these interpretations. 

\bigskip

\textbf{Ordered partitions acted upon by cyclic permutations.} Consider the set of all ordered partitions of a positive integer $n$ into $k$ parts acted upon by the cyclic permutation $(1 2 \ldots k)$. Razen, Seberry, and Wehrhahn \cite{RSW} obtained explicit formulas for the cardinality of the resulting family of orbits and for the number of orbits in this family having exactly $k$ elements. These formulas coincide with the expressions for $P_n(k,0)$ and $P_n(k,1)$, respectively, when $n$ or $k$ is odd (see Corollary~\ref{special cases: b=0,1}). Razen et al. \cite{RSW} also discussed an application in coding theory in finding the complete weight enumerator of a code generated by a circulant matrix.

\bigskip

\textbf{Permutations with given cycle structure and descent set.} Gessel and Reutenauer \cite{GERE} counted permutations in the symmetric group $S_n$ with a given cycle structure and descent set. One of their results gives an explicit formula for the number of $n$-cycles with descent set $\lbrace k \rbrace$, which coincides with the expression for $P_n(k,1)$ when $n$ or $k$ is odd.

\bigskip

\textbf{Fixed-density necklaces and Lyndon words.} If $n$ or $k$ is odd then the expressions for $P_n(k,0)$ and $P_n(k,1)$ give, respectively, the number of fixed-density binary necklaces and fixed-density binary Lyndon words of length $n$ and density $k$, as described by Gilbert and Riordan \cite{GIRI}, and Ruskey and Sawada \cite{RUSA}.

\bigskip

\textbf{Necklace polynomial.} The function $T_n(b)$ is closely related to the polynomial

$$
M(q, n)= \frac{1}{n}\sum_{d\, \mid \, n}\mu(d)q^{\frac{n}{d}},
$$
which is called the \textit{necklace polynomial} of degree $n$ (it is easy to see that $M(q, n)$ is integer-valued for all $q \in \Z$). In fact, if $n$ is odd then $M(2, n)=T_n(1)$. The necklace polynomials turn up in various contexts in combinatorics and algebra.

\bigskip

\textbf{Quasi-necklace polynomial.} The function $T'_n(b)$ is also closely related to the polynomial

$$
M'(q, n)= \frac{1}{2n}\sum_{d\, \mid \, n}\mu(d)q^{\frac{n}{d}},
$$
that we call the \textit{quasi-necklace polynomial} of degree $n$. In fact, if $n$ is odd then $M'(2, n)=T'_n(1)$. The quasi-necklace polynomials also turn up in various contexts in combinatorics. For example, they appear as:

\begin{itemize}

   \item the number of transitive unimodal cyclic permutations obtained by Weiss and Rogers \cite{WERO} (motivated by problems related to the structure of the set of periodic orbits of one-dimensional dynamical systems) using methods related to the work of Milnor and Thurston \cite{MITH}. See also \cite{THIB} which gives a generating function for the number of unimodal permutations with a given cycle structure;
   
   \item the number of periodic patterns of the tent map \cite{AREL}.
   
\end{itemize}

\section*{Acknowledgements}
The authors would like to thank the anonymous referees for a careful reading of the paper and helpful suggestions. During the preparation of this work the first author was supported by a Fellowship from the University of Victoria (UVic Fellowship).

\end{document}